\documentclass[conference,a4paper]{IEEEtran}

\usepackage{pgfplots}
\pgfplotsset{compat=newest}
\usepackage[innermargin=0.545in,outermargin=0.65in,top=0.545in,bottom=.7in]{geometry}
\usepackage{tikz}
\usepackage{amsmath}
\usepackage{amsthm}
\usepackage{amsfonts}
\usepackage{amssymb}
\usepackage{array}
\usepackage[utf8]{inputenc}
\usepackage{graphicx} 
\usepackage{algorithm}
\usepackage{algorithmic}
\usepackage[T1]{fontenc}
\usepackage[utf8]{inputenc}
\usepackage{authblk}
\usepackage{xcolor}
\usepackage{enumitem}

\newtheorem{theorem}{Theorem}

\newtheorem{proposition}{Proposition}
\newtheorem{corollary}{Corollary}
\newtheorem{lemma}{Lemma}

\theoremstyle{definition}
\newtheorem{definition}{Definition}

\newtheorem{example}{Example}
\newtheorem{remark}{Remark}

\newcommand{\C}{\mathcal{C}}

\newcommand{\GL}{\mathrm{GL}}
\newcommand{\F}{\mathbb{F}}

\newcommand{\rk}{\mathrm{rk}}

\newcommand{\supp}{\mathrm{supp}}

\newcommand{\Aut}{\mathrm{Aut}}

\newcommand{\Fm}{\mathbb{F}_{q^m}}

\newcommand{\Fq}{\mathbb{F}_{q}}

\IEEEoverridecommandlockouts

\title{Invariants and Inequivalence of \\ Linear Rank-Metric Codes\thanks{This work was supported by SNF grant no.~169510 and by the German Israeli Project Cooperation (DIP) grant no.~KR3517/9-1.}}

\author[1]{Alessandro Neri}
\author[2]{Sven Puchinger}
\author[3]{Anna-Lena Horlemann-Trautmann}

\affil[1]{Department of Mathematics, University of Z{u}rich, Switzerland}
\affil[2]{Institute for Communications Engineering, Technical University of Munich (TUM), Germany}
\affil[3]{Faculty of Mathematics and Statistics, University of St. Gallen, Switzerland}

\newcommand{\Gal}{\mathrm{Gal}}
\newcommand{\NN}{\mathbb{N}}
\newcommand{\autom}{\sigma}

\newcommand{\Fqm}{\mathbb{F}_{q^m}}
\newcommand{\Fqr}{\mathbb{F}_{q^r}}
\newcommand{\Fqn}{\mathbb{F}_{q^n}}
\newcommand{\Code}{\mathcal{C}}
\def\ve#1{{\mathchoice{\mbox{\boldmath$\displaystyle #1$}}%
              {\mbox{\boldmath$\textstyle #1$}}%
              {\mbox{\boldmath$\scriptstyle #1$}}%
              {\mbox{\boldmath$\scriptscriptstyle #1$}}}}
\newcommand{\numTwists}{\ell}
\newcommand{\tVec}{{\ve t}}
\newcommand{\hVec}{{\ve h}}
\newcommand{\etaVec}{{\ve \eta}}
\newcommand{\alphaVec}{{\ve \alpha}}
\newcommand{\betaVec}{{\ve \beta}}

\newcommand{\Saut}{\mathcal S^{\autom}}

\newcommand{\CGab}[3]{\Code_{#1}^{#2}(#3)}
\newcommand{\CTGabSheekey}[3]{\Code_{#1}^{#2}(#3)}
\newcommand{\CTGab}[3]{\Code_{#1}^{#2}(#3)}
\newcommand{\CGabNew}[3]{\Code_{#1}^{#2}(#3)}

\newcommand{\A}{\ve{A}}
\newcommand{\G}{\ve{G}}

\renewcommand{\u}{\ve{u}}
\renewcommand{\v}{\ve{v}}
\newcommand{\dR}{\mathrm{d}_\mathrm{R}}

\begin{document}
\setlist[enumerate]{leftmargin=5.5mm}
\setlist[itemize]{leftmargin=5.5mm}
\renewcommand{\baselinestretch}{0.93}

\maketitle

\begin{abstract}
We show that the sequence of dimensions of the linear spaces, generated by a given rank-metric code together with itself under several applications of a field automorphism, is an invariant for the whole equivalence class of the code. These invariants give rise to an easily computable criterion to check if two codes are inequivalent. With this criterion we then derive bounds on the number of equivalence classes of classical and twisted Gabidulin codes.
\end{abstract}


\section{Introduction}
Rank-metric codes have received a lot of attention in the last decades, due to their applications in network coding, distributed storage and post-quantum cryptography. Optimal rank-metric codes, achieving the Singleton bound, are called \emph{maximum rank distance (MRD) codes.} The most prominent family of these codes, called Gabiulin codes,  was introduced, independently, by Delsarte \cite{Delsarte_1978}, Gabidulin \cite{Gabidulin_TheoryOfCodes_1985}, and Roth \cite{Roth_RankCodes_1991}. Several generalizations of this construction have been made afterwards, e.g. in \cite{kshevetskiy2005new}.
Since it was shown in \cite{neri2018genericity} that, for a sufficiently large extension field degree $m$,  there are plenty of MRD codes that are not (generalized) Gabidulin codes, it has been an ongoing research problem to find non-Gabidulin MRD codes.

A first answer in this direction was given by Sheekey in \cite{sheekey2015new}. His \emph{twisted Gabidulin codes} were the first construction of a family of non-Gabidulin MRD codes. At the same time a special case of them was also discovered in \cite{otal2016explicit}. This family has been generalized 
 in \cite[Remark~9]{sheekey2015new}, and in \cite{lunardon2015generalized,puchinger2017further,gabidulin2017new}.  
 Moreover, other non-Gabidulin MRD codes have been constructed; the references of which are ommitted due to space restrictions.
 
When looking for new MRD code constructions, it is important to understand if the new codes are equivalent to any of the already known codes.
Furthermore, knowing the equivalence classes of codes is also of interest in code-based cryptography. For instance, many attacks and distinguishers on a code can be directly applied to any other code in its equivalence class.

Hence, it is helpful to have easily computable criteria to check if codes belong to a given family. For (generalized) Gabidulin codes such a criterion, based on the dimension of the intersection of the code with itself under some field automorphism, was given in \cite{horlemann2015new}. In this work we generalize this result to more than just one application of the field automorphism, and derive results about the number of equivalence classes of certain MRD codes.
Compared to the criterion in \cite{horlemann2015new}, the new method is able to distinguish more classes of rank-metric codes from each other.

\section{Preliminaries}

\subsection{Rank-Metric Codes}
Let $q$ be a prime power,  
 $\Fq$ the finite field of cardinality $q$, and  $\Fm$ the extension field of degree $m$ of $\Fq$. In this work we consider codes in the vector space $\F_{q^m}^n$, equipped with the rank metric $\dR$, defined as
$\dR(\u,\v):=\dim_{\Fq}(\supp_q(\u-\v))$
for any $\u,\v \in \Fqm^n$, where 
$\supp_q(\u) := \langle u_1,\ldots, u_n \rangle_{\Fq}$ is  the \emph{$\Fq$-support} of $\u$.

A \emph{rank-metric code} is a subset $\C\subseteq \Fm^n$ endowed with the rank distance $\dR$. Such a rank-metric code is called \emph{$\Fm$-linear} if it is a subspace of $\Fm^n$. 
In this work we will only deal with rank-metric codes that are linear over $\Fm$.
 The integer $n$ is called the \emph{length} of $\C$, while the \emph{dimension} of $\C$ is the integer $k=\dim_{\Fm}(\C)$. The \emph{minimum distance} of $\C$ is $d(\C)=\min\{\dR(u,v) \mid u,v \in \C, u \neq v\}$. 
For a $k$-dimensional $\Fm$-linear rank metric code of length $n$ and minimum rank distance $d$, the Singleton bound \cite{Delsarte_1978,Gabidulin_TheoryOfCodes_1985,Roth_RankCodes_1991} states that 
$d\leq n-k+1$. 
Codes that meet this bound 
with equality are called \emph{maximum rank distance (MRD) codes}. It is well-known that $\Fqm$-linear MRD codes exist for any parameters if and only if $n\leq m$. For this reason, we will assume $n\leq m$ in this work.

\subsection{Equivalence of Rank-Metric Codes}

We say that two rank-metric codes in $\Fm^n$ are equivalent if there exists a rank isometry that maps one code onto the other.
Equivalent codes share many important  properties.
In the Hamming metric, it is known that deciding whether two codes are equivalent is a hard (but not NP-hard) problem \cite{petrank1997code}.

Denote by $\mathrm{Aut}(\F_{q^m})$ the \emph{automorphism group} of
$\F_{q^m}$, and by $\GL_n(q):=\{A\in \Fq^{n\times n} \mid \rk (A) =n\}$
 the \emph{general linear group} of degree $n$ over $\Fq$.
The \emph{semi-linear rank isometries}  on $\F_{q^m}^{n}$ (i.e., the distance-preserving mappings that are linear up to a field automorphism) are induced by the
isometries on $\Fq^{m\times n}$ and are hence well-known, see e.g.\
\cite{morrison2014equivalence}:
\begin{lemma}\cite[Proposition~2]{morrison2014equivalence}\label{isometries}
  The semilinear $\Fq$-rank isometries on $\Fm^{n}$ are of the
  form
  \[(\lambda, \A, \theta) \in \left( \Fm^* \times \GL_n(q) \right)
  \rtimes \mathrm{Aut}(\Fm) ,\] acting on $ \Fm^n \ni
  (v_1,\dots,v_n)$ via
  \[(v_1,\dots,v_n) (\lambda, \A, \theta) = (\theta(\lambda
  v_1),\dots,\theta(\lambda v_n)) \A .\] 
\end{lemma}

\begin{definition}
Two $\Fm$-linear rank-metric codes $\C, \C' \subseteq \Fm^n$ are \emph{semi-linearly equivalent} if there exists $(\lambda, \A, \theta) \in \left( \Fm^* \times \GL_n(q) \right)
  \rtimes \mathrm{Aut}(\Fm)$ such that  $\mathcal{C}' = \theta(\lambda \mathcal{C}) \A$.
\end{definition}
 In the sequel we can always assume $\lambda=1$, since a multiplication of $\C$ by a scalar in $\Fm$ does not change the code.
We remark that there exists a different notion of rank equivalence, when one considers rank-metric codes as $\Fq$-subspaces of $\Fq^{m\times n}$. We will not deal with this notion in this work, however, it is worth mentioning that for $\Fm$-linear codes the two notions coincide (see \cite[Proposition 2.5]{sheekey2018rank}).

\subsection{Known $\Fqm$-Linear Constructions}

\begin{definition}\label{def:codes}
Let $\alphaVec \in \Fqm^n$ be a vector whose entries $\alpha_i$ are linearly independent over $\Fq$, $1 \leq k \leq n$, and $\autom$ a generator of $ \Gal(\Fqm/\Fq)$.
\begin{enumerate}
\item A \emph{Gabidulin code} \cite{Delsarte_1978,Gabidulin_TheoryOfCodes_1985,Roth_RankCodes_1991} is defined by
\begin{align*}
\CGab{k}{\autom}{\alphaVec} := \left\langle \alphaVec, \autom(\alphaVec), \dots, \autom^{k-1}(\alphaVec) \right\rangle_{\Fqm}.
\end{align*}
\item Let $\eta \in \Fqm \setminus \{0\}$. The corresponding \emph{(Sheekey's) twisted Gabidulin code} \cite{sheekey2015new} is defined by
\begin{align*}
\CTGabSheekey{k, \eta}{\autom}{\alphaVec} := \left\langle \alphaVec+\eta \autom^k(\alphaVec), \autom(\alphaVec), \dots, \autom^{k-1}(\alphaVec) \right\rangle_{\Fqm}.
\end{align*}
\item Choose $\numTwists \in \NN$, which we call the \emph{number of twists}. Let $\hVec \in \{0,\dots,k-1\}^\numTwists$ and $\tVec \in \{1,\dots,n-k\}^\numTwists$ such that the $h_i$ are distinct and the $t_i$ are distinct. Furthermore, let $\etaVec \in (\Fqm \setminus \{0\})^\numTwists$. The corresponding \emph{(generalized) twisted Gabidulin code} \cite{puchinger2017further} is defined by
\begin{align*}
\CTGab{k,\tVec,\hVec,\etaVec}{\autom}{\alphaVec} := \big\langle\! \left\{\autom^{h_i}(\alphaVec)+\eta_i \sigma^{k-1+t_i}(\alphaVec) : i=1,\dots,\numTwists \right\} \\ 
\cup \left\{\autom^{i}(\alphaVec) : i \in \{0,\dots,k-1\} \setminus \{h_1,\dots,h_\ell\} \right\} \!\big\rangle_{\Fqm}.
\end{align*}
\item Let $\eta \in \Fqm \setminus \{0\}$ and restrict $k$ to satisfy $m-k \leq k$. Then, \emph{Gabidulin's new codes} \cite{gabidulin2017new} are defined by
\begin{align*}
\CGabNew{k,\eta}{\autom,\mathsf{new}}{\alphaVec}\! :=\! \big\langle \!\left\{ \sigma^{i}(\alphaVec)+\sigma^{i}(\eta) \sigma^{k+i}(\alphaVec) : 0\leq i < m-k \right\} \\
\cup \left\{ \sigma^i(\alphaVec) : m-k \leq i <k \right\} \!\big\rangle_{\Fqm}.
\end{align*}
\end{enumerate}
\end{definition}

Gabidulin codes are MRD for any choice of $\alphaVec$ \cite{Delsarte_1978,Gabidulin_TheoryOfCodes_1985,Roth_RankCodes_1991}.
Sheekey's twisted Gabidulin codes and Gabidulin's new codes are MRD if $\eta$ has norm $N(\eta) \neq (-1)^{km}$ \cite{sheekey2015new,gabidulin2017new}.
Twisted Gabidulin codes are a generalization of Sheekey's codes.
In general, there is a sufficient MRD condition if the $\alpha_i$ are chosen from a subfield $\Fqr \subseteq \Fqm$ with $r 2^\numTwists \mid m$ and a suitable choice of the $\eta_i$ \cite{puchinger2017further} (see also \cite[Chapter~7]{puchinger2018construction} for more details). Note that this gives codes of length $n \leq 2^{-\numTwists}m$. It is an open problem whether longer MRD codes exist for arbitrary $\tVec$ and $\hVec$.
In the special case $\numTwists=1$, we write $t := \tVec = t_1 \in \NN$ and $h := \hVec = h_1 \in \NN_0$.
Gabidulin's new codes can be seen as a special case of (generalized) twisted codes with
\begin{equation*}
\ell = m-k, \quad h_i = i-1, \quad t_i = i, \quad \text{and} \quad \eta_i = \sigma^{i-1}(\eta)
\end{equation*}
for $i=1,\dots,m-k$.

Note that all codes defined above are rank-metric codes of length $n$ and dimension $k$.
The dimension directly follows from the fact that $\alphaVec,\sigma(\alphaVec),\dots,\sigma^{i-1}(\alphaVec)$ are the rows of the \emph{$\autom$-Moore matrix}
\begin{align*}
\setlength{\arraycolsep}{2pt}
\def\arraystretch{0.6}
M_{i}^{\autom}(\alphaVec) :=
\begin{pmatrix}
\alphaVec \\
\autom(\alphaVec) \\
\vdots \\
\autom^{i-1}(\alphaVec) \\
\end{pmatrix}
=
\begin{pmatrix}
\alpha_1 & \cdots & \alpha_n \\
\autom(\alpha_1) & \cdots & \autom(\alpha_n) \\
\vdots & \ddots & \vdots \\
\autom^{i-1}(\alpha_1) & \cdots & \autom^{i-1}(\alpha_n)
\end{pmatrix},
\end{align*}
which are linearly independent over $\Fqm$ since the entries of $\alphaVec$ are linearly independent over $\Fq$. This is a direct consequence of the following statement, which is a generalization of \cite[Corollary 2.38]{lidl1997finite} and directly follows from \cite[Corollary 4.13]{lam1988vandermonde}.
\begin{proposition}\label{prop:rankMoore}
Let $\alphaVec \in \Fm^n$. Then $\rk(M_{i}^{\autom}(\alphaVec))=\min\{i,r\}$, where $r=\dim_{\Fq}\langle \alpha_1,\ldots, \alpha_n\rangle_{\Fq}$.
\end{proposition}

\section{New Invariants in the Rank Metric}

\subsection{General Results}
In recent works it has been noticed that the dimension of the code $\C + \sigma(\C)$ is an invariant of a code $\C$ under rank-metric equivalences. In \cite{horlemann2015new} this was used to derive a criterion for checking whether a given code is Gabidulin or not. In \cite{puchinger2018construction} it was used to show that some twisted Gabidulin codes are inequivalent to known constructions. We generalize this invariant here.

\begin{lemma}\label{lem:invariant}
 Let $0<s_1<\ldots< s_r<m$ be positive integers, and let $\C_1,\C_2$ be two equivalent $\Fm$-linear rank metric codes. Then, $\C_1' := \C_1+\autom^{s_1}(\C_1)+\ldots+\autom^{s_r}(\C_1)$ is equivalent to $\C_2' := \C_2+\autom^{s_1}(\C_2)+\ldots+\autom^{s_r}(\C_2)$.
In particular, $\dim \C_1' = \dim \C_2'$.
\end{lemma}

\begin{proof}
 Since $\C_1$ and $\C_2$ are equivalent, there exist $\theta \in \Aut(\Fm)$, $\A\in \GL_n(\Fq)$ such that
$\C_1=\theta(\C_2)\A$. Therefore,
\begin{align*}
&\C_1+\autom^{s_1}(\C_1)+\ldots+\autom^{s_r}(\C_1) \\
&=\theta(\C_2)\A+\autom^{s_1}(\theta(\C_2))\autom^{s_1}(\A)+\ldots+\autom^{s_r}(\theta(\C_2))\autom^{s_r}(\A) \\
&\stackrel{(*)}{=} \theta(\C_2)\A+\theta(\autom^{s_1}(\C_2))\A+\ldots+\theta(\autom^{s_r}(\C_2))\A  \\
&=\theta(\C_2+\autom^{s_1}(\C_2)+\ldots+\autom^{s_r}(\C_2))\A,
\end{align*}
where (*) holds because the entries of $\A$ belong to $\Fq$ and $\Aut(\Fm)$ is an abelian group.
\end{proof}

Lemma~\ref{lem:invariant} implies that if two $\Fqm$-linear codes $\C_1,\C_2$ have different dimension of $\C_1'$ and $\C_2'$, then they must be inequivalent. Hence, checking the dimensions of $\C_1'$ and $\C_2'$ for different choices of the powers $s_i$ gives a sufficient condition for codes to be inequivalent.
It is notable that computing the dimension of $\C+\autom^{s_1}(\C)+\ldots+\autom^{s_r}(\C)$ of a code $\C$ with generator matrix $\G$ can be done by computing the rank of the $(r+1)k \times n$ matrix $(\G^\top,\sigma^{s_1}(\G)^\top, \dots, \sigma^{s_r}(\G)^\top)^\top$,
which costs at most $O(\max\{r^2k^2n,n^2 r k\})$ field operations.
In the following, we restrict to the special case of consecutive $s_i = i$ for the sake of easier proofs.

\begin{definition}
Let $\C$ be a code of length $n$ and dimension $k$, and let $\autom \in \Gal(\Fm/\Fq)$. We consider the \emph{$\autom$-sequence of codes} $\{\mathcal S^\autom_i(\C)\}_{i=0}^m$,  defined as the sum of $\Fm$-vector spaces
\begin{equation*}
\mathcal S^\autom_i(\C)=\sum_{j=0}^i \autom^j(\C),
\end{equation*}
 and the associate sequence of integers $\{s^\autom_i(\C)\}_{i=0}^m$ given by the dimensions of the codes, i.e.
$$s^\autom_i(\C)=\dim (\mathcal S^\autom_i(\C)).$$
\end{definition}

As a consequence of Lemma \ref{lem:invariant}, we get that the sequence $\{s^\autom_i(\C)\}$ is an invariant of linear rank metric codes, i.e., it is stable under code equivalence.

\begin{proposition}\label{prop:propertiesSi}
Let $\C\subseteq \Fm^n$ be a $\Fm$-linear rank metric code of dimension $k$. Then:
\begin{enumerate}
\item $\{s^\autom_i(\C)\}$ is a non-decreasing sequence of integers between $k$ and $n$.
\item $\mathcal S^\autom_{i+j}(\C)=\mathcal S^\autom_i(\mathcal S^\autom_j(\C))$.
\item $s^\autom_i(\C)=s^\autom_{i+1}(\C)$ if and only if $\mathcal S^\autom_i(\C)$ has a basis of elements in $\Fq^n$, i.e., a generator matrix over $\Fq$.
\item If $s^\autom_i(\C)=s^\autom_{i+1}(\C)$ then $s^\autom_{i+j}(\C)=s^\autom_i(\C)$ for all $j \geq 0$.
\item $s^\autom_{n-k}(\C)=s^\autom_{n-k+j}(\C)$ for all $j \geq 0$.
\item $s^\autom_{i+j+1}(\C)-s^\autom_{i+j}(\C)\leq s^\autom_{i+1}(\C)-s^\autom_{i}(\C)$ for all $j\geq 0$.
\item $s^\autom_{i+1}(\C)-s^\autom_{i}(\C)\leq k$ for all $i\geq 0$.
\end{enumerate}
\end{proposition}

\begin{proof}
\begin{enumerate}
\item Follows from $\Saut_i(\C)\subseteq \Saut_{i+1}(\C) \subseteq \Fm^n$.
\item It holds that $\Saut_i(\Saut_j(\C))=\sum_{\ell=0}^i\autom^\ell(\Saut_j(\C))=\sum_{\ell=0}^i\sum_{r=0}^j \autom^{\ell+r}(\C)=\sum_{s=0}^{i+j}\sigma^s(\C)=\Saut_{i+j}(\C)$.
\item Suppose $s^\autom_i(\C)=s^\autom_{i+1}(\C)$, then $\Saut_i(\C)=\Saut_{i+1}(\C)$, and by part 2, we get $\Saut_1(\Saut_i(\C))=\Saut_i(\C)$. This is true if and only if $\autom(\Saut_i(\C))=\Saut_i(\C)$, and we can conclude using \cite[Lemma 4.5]{horlemann2015new}.
\item $s^\autom_i(\C)=s^\autom_{i+1}(\C)$ implies that $\autom(\Saut_i(\C))=\Saut_i(\C)$, and therefore, $\Saut_{i+j}(\C)=\autom^j(\Saut_i(\C))=\Saut_i(\C)$
 for all $j \geq 0$. 
\item Given a $k$-dimensional code $\C$, let $r^\autom(\C)=\min\{i \mid s^\autom_i(\C)=s^\autom_{i+1}(\C) \}$. If $r^\autom(\C)\leq n-k$, then by part 4 we can conclude. Suppose by contradiction that $r(\C)>n-k$. Then we get a chain 
$k=s^\autom_0(\C)<s^\autom_1(\C)< \ldots <s^\autom_{n-k}(\C)<s^\autom_{n-k+1}(\C).$
This implies that $s^\autom_i(\C) \geq k+i$, and in particular $s^\autom_{n-k+1}(\C)\geq k+n-k+1=n+1$, but this is impossible since $\Saut_{n-k+1}(\C)\subseteq \Fm^n$.
\item We prove it for $j=1$, and the general case follows then by induction. Suppose $s^\autom_{i+1}(\C)=s^\autom_{i}(\C)+\ell$. Then $\dim (\Saut_i(\C)+\autom(\Saut_i(\C)))=\dim (\Saut_i(\C))+\ell$. This implies that $\autom(\Saut_i(\C))=W+U$, where $W\subseteq \Saut_i(\C)$, $U\cap \Saut_i(\C)=\{0\}$ and $\dim U=\ell$. Hence, $\Saut_{i+2}(\C)=\Saut_1(\Saut_{i+1}(\C))=\Saut_{i+1}(\C)+\autom(\Saut_{i+1}(\C))=\Saut_i(\C)+U+\autom(\Saut_i(\C))+\autom(U)$. However, $U \subseteq \autom(\Saut_i(\C))$, and therefore $\Saut_{i+2}(\C)=\Saut_i(\C)+\autom(\Saut_i(\C))+\autom(U)=\Saut_{i+1}(\C)+\autom(U)$. Since $\dim \autom(U)=\dim(U)=\ell$, we conclude. 
\item We have $\Saut_1(\C)=\C+\autom(\C)$ and thus $s^\autom_1(\C)=\dim(\C+\autom(\C))\leq \dim(\C)+\dim(\autom(\C))=s^\autom_0(\C)+k$. The statement follows with part 6. \qedhere
\end{enumerate}
\end{proof}

\begin{remark}
The sequence $\{\mathcal{S}_i^{\autom}(\C)\}_{i=0}^{m}$ was already considered in \cite{overbeck2008structural}, and following work, to retrieve the structure of a Gabidulin code (i.e., the vector $\alphaVec$) from an obfuscated generator matrix thereof, which led to an efficient attack on a cryptosystem based on Gabidulin codes. For the same purpose, it was shown in \cite[Proposition~2]{coggia2019security} that for a random $k$-dimensional code $\C$, we have $s^\autom_i(\C) = \min\{n,(i+1)k\}$ with high probability.
 To the best of our knowledge, it has so far not been used to study inequivalences of rank-metric codes. 
\end{remark}

In the following subsections, we explicitly compute the sequences $\{s^\autom_i(\C)\}_{i=0}^m$ some code families. 
In some cases this implies the exact number of pairwise inequivalent codes in the respective code family.

\subsection{The Sequence for Gabidulin Codes}

\begin{proposition}\label{prop:GabidulinSeq}
Let $\C := \CGab{k}{\sigma}{\alphaVec}$ be a Gabidulin code of length $n$ and dimension $k$, as defined in Definition~\ref{def:codes},
and $i \in \mathbb N$. 
\begin{enumerate}
\item If $0 \leq \ell \leq k$, then $\mathcal S^{\autom^{\ell}}_i(\C)=\Saut_{i\ell}(\C)= \CGab{k+i\ell}{\autom}{\alphaVec}$ and $s^{\autom^\ell}_i(\C)=\min\{k+i\ell,n\}$.
\item If $m-k \leq \ell <m$, then $\mathcal S^{\autom^{\ell}}_i(\C)= \CGab{k+i(m-\ell)}{\autom^{-1}}{\autom^{k-1}(\alphaVec)}$ and $s^{\autom^\ell}_i(\C)=\min\{k+i(m-\ell),n\}$.
\item If $m=n$, and $0 \leq \ell \leq m-1$ , then $s^{\autom^\ell}_1(\C)=\min\{r,n\}$, where
$$r=\begin{cases} k+\ell  & \mbox{ if }  0\leq \ell \leq k\\
 k+m-\ell & \mbox{ if } m-k\leq \ell \leq m-1\\
2k & \mbox{ if } k+1 \leq \ell \leq m-k-1 \end{cases} .$$
\end{enumerate}
\end{proposition}

\begin{proof}
\begin{enumerate}
\item Let $0\leq \ell\leq k$. Then $\mathcal S^{\autom^\ell}_i(\C)=\langle \alphaVec,\ldots,\autom^{k-1}(\alphaVec),\autom^k(\alphaVec),\ldots, \autom^{k+i\ell-1}(\alphaVec)\rangle = \CGab{k+i\ell}{\autom}{\alphaVec}$. The computation of $s^\autom_i(\C)$ follows from Proposition \ref{prop:rankMoore}.
\item  If $m-k \leq \ell \leq m-1$, then the claim follows considering that $\CGab{k}{\sigma}{\alphaVec} = \CGab{k}{\sigma^{-1}}{\autom^{k-1}(\alphaVec)}$, and applying $\autom^\ell=(\autom^{-1})^{m-\ell}$, with $0\leq m-\ell \leq k$. The computation of $s^\autom_i(\C)$ follows again from Proposition \ref{prop:rankMoore}.
\item If $0\leq \ell\leq k$ or $m-k \leq \ell \leq m-1$, it follows from part 1. On the other hand, if $ k+1 \leq \ell \leq m-k-1$, then 
$\mathcal S^{\autom^\ell}_1(\C)=\langle \alphaVec, \ldots, \autom^{k-1}(\alphaVec),\autom^{\ell}(\alphaVec),\ldots, \autom^{k+\ell-1}(\alphaVec)\rangle$, and  the generators of $\mathcal S^{\autom^\ell}_1(\C)$ are all distinct rows of the Moore matrix $M_{m}^{\sigma}(\alphaVec)$. Hence, we conclude using Proposition \ref{prop:rankMoore}. \qedhere
\end{enumerate}
\end{proof}

\begin{theorem}\label{thm:numberGabidulin}
Let $1<k<n-1$, and $m=n$. Then there are exactly $\frac{\phi(m)}{2}$ inequivalent Gabidulin codes of length $m$ and dimension $k$, where $\phi$ denotes Euler's totient function.
\end{theorem}

\begin{proof}
Let $\autom, \tau$ be two generators of $\Gal(\Fm/\Fq)$, and $\alphaVec,\betaVec\in \Fm^m$ be two vectors with entries that are linearly independent over $\Fq$. Consider two Gabidulin codes $\C=\CGab{k}{\sigma}{\alphaVec}$ and $\C'=\CGab{k}{\tau}{\betaVec}$. First we show that, if $\tau=\autom$ then $\C$ and $\C'$ are equivalent. Since $n=m$, then $\supp_q(\alphaVec)=\supp_q(\betaVec) = \Fqm$. Therefore, there exists $\A \in \GL_m(q)$ such that $\alphaVec \A = \betaVec$. This implies that $\C\cdot \A=\C'$ and they are equivalent. 
Now, suppose that $\tau=\autom^{-1}$. Then, by the first part we can assume $\alphaVec = \betaVec$. Since $\CGab{k}{\sigma^{-1}}{\betaVec} = \CGab{k}{\sigma}{\autom^{m-k+1}(\betaVec)}$, we obtain again that they are equivalent. 
Finally, if  $\tau\notin \{\autom, \autom^{-1}\}$, then $\tau=\autom^\ell$, with $\ell \notin \{1,-1\}$, and by part 3 of Proposition \ref{prop:GabidulinSeq}, we have $s^\tau_1(\C')=k+1$ and $s^\tau_1(\C)\geq k+2$. By Lemma \ref{lem:invariant} we deduce that they cannot be equivalent. Since there are exactly $\phi(m)$ generators for $\Gal(\Fm/\Fq)$, we conclude.
\end{proof}

\subsection{The Sequence for Sheekey's Twisted Gabidulin Codes}

\begin{proposition}\label{prop:TwGabidulinSeq}
Let $\C := \CTGabSheekey{k,\eta}{\autom}{\alphaVec}$ be one of Sheekey's twisted Gabidulin codes as defined in Definition~\ref{def:codes}, where $\eta \in \Fm^*$ with $N(\eta) \neq (-1)^{km}$.
\begin{enumerate}
\item If $1 \leq \ell \leq k-1$, then $\mathcal S^{\autom^{\ell}}_i(\C) = \CGab{k+i\ell+1}{\sigma}{\alphaVec}$ and $s^{\autom^\ell}_i(\C)=\min\{k+i\ell+1,n\}$.
\item If $m-k+1 \leq \ell <m$, then $\mathcal S^{\autom^{\ell}}_i(\C) = \CGab{k+i(m-\ell)+1}{\autom^{-1}}{\autom^{k}(\alphaVec)}$ and $s^{\autom^\ell}_i(\C)=\min\{k+i(m-\ell)+1,n\}$.
\item If $m=n$, and $1 \leq \ell \leq m-1$ , then $s^{\autom^\ell}_1(\C)=\min\{r,m\}$, where
\vspace{-0.2cm}
$$r=\begin{cases} k+\ell+1  & \mbox{ if }  1\leq \ell \leq k-1\\
 k+m-\ell+1& \mbox{ if } m-k+1\leq \ell \leq m-1\\
2k & \mbox{ if } k \leq \ell \leq m-k \end{cases} .$$
\end{enumerate}
\end{proposition}

\begin{proof}
\begin{enumerate}
\item Let $1\leq \ell\leq k-1$ and $i=1$. Then $\mathcal S^{\autom^\ell}_1(\C)=\langle \alphaVec+\eta \autom^k(\alphaVec),\ldots,\autom^{k-1}(\alphaVec),\autom^\ell(\alphaVec)+\autom^\ell(\eta)\autom^{k+\ell}(\alphaVec), \autom^{\ell+1}(\alphaVec)\ldots, \autom^{k+\ell-1}(\alphaVec)\rangle\subseteq \CGab{k+\ell}{\autom}{\alphaVec}$. Moreover, $\mathcal S^{\autom^\ell}_1(\C) \supseteq \{\autom(\alphaVec),\ldots, \autom^{\ell+k-1}(\alphaVec)\}$.
Furthermore, it contains $\alphaVec+\eta\autom^k(\alphaVec)$ and $\autom^\ell(\alphaVec)+\autom^\ell(\eta)\autom^{k+\ell}(\alphaVec)$.
Since $1\leq \ell \leq k-1$, then it contains  $\autom^k(\alphaVec)$, and $\autom^{\ell}(\alphaVec)$, and therefore  $\alphaVec, \autom^{k+\ell}(\alphaVec) \in \mathcal S^{\autom^\ell}_1(\C)$, and we deduce $\mathcal S^{\autom^\ell}_1(\C)= \CGab{k+\ell+1}{\autom}{\alphaVec}$.
If $i >1$, by part 2 of Proposition~\ref{prop:propertiesSi}, we have $\mathcal S^{\autom^\ell}_i(\C)=\mathcal S^{\autom^\ell}_{i-1}(\mathcal S^{\autom^\ell}_1(\C))$, and we conclude using part 1 of Proposition \ref{prop:GabidulinSeq}. 
\item Observe that $\CTGabSheekey{k,\eta}{\autom}{\alphaVec} = \CTGabSheekey{k,\eta^{-1}}{\autom^{-1}}{\autom^k(\alphaVec)}$, and that $\autom^{\ell}=(\autom^{-1})^{m-\ell}$, with $1\leq m-\ell \leq k-1$. Then, the claim follows from part 1.
\item If $0\leq \ell\leq k$ or $m-k \leq \ell \leq m-1$, it follows from part~1. The case $k \leq \ell \leq m-k$ is omitted, due to lack of space. \qedhere
\end{enumerate}
\end{proof}

For a prime power $q$, and two integers $k,m$ we consider the left group action of $\Aut(\Fm)$ on the set $\{\alpha \in \Fm \mid N(\alpha)\neq (-1)^{km} \}$, given by $\theta \cdot \alpha:=\theta(\alpha)$. We denote by $\mathcal X_q(m,k)$ the set of orbits of this group action. Observe that the action above is well-defined. Indeed, if $N(\alpha) \neq (-1)^{km}$, and $\theta \in \Aut(\Fm)$, then
$N(\theta(\alpha))=\theta(N(\alpha))$. This is due to the fact that $\Aut(\Fm)$ is a cyclic group, and it contains $\Gal(\Fm/\Fq)$.
Since $(-1)^{km}$ belongs to the prime field, and therefore is fixed by any automorphism $\theta \in \Aut(\Fm)$, we have that also $\theta(N(\alpha))\neq (-1)^{km}$.

\begin{theorem}
Let $2<k<n-2$, and $m=n$. Then there are exactly $\frac{\phi(m)}{2}|\mathcal X_q(m,k)|$ inequivalent Sheekey's twisted Gabidulin codes of length $m$, dimension $k$, and $N(\eta) \neq (-1)^{km}$ (i.e., MRD), where $\phi$ denotes Euler's totient function.
\end{theorem}

\begin{proof}
 The proof is similar to the one of Theorem \ref{thm:numberGabidulin}.  Let $\autom, \tau$ be two generators of $\Gal(\Fm/\Fq)$,  $\alphaVec, \betaVec \in \Fm^m$ be two vectors with entries that are linearly independent over $\Fq$ and let $\eta, \eta' \in \Fm$ with norm not equal to $(-1)^{km}$. Consider two twisted Gabidulin codes $\C=\CTGabSheekey{k,\eta}{\autom}{\alphaVec}$ and $\C'=\CTGabSheekey{k,\eta'}{\tau}{\betaVec}$.
Suppose  $\tau\notin\{ \autom, \autom^{-1}\}$, 
then $\tau=\autom^\ell$, with $\ell \notin\{ 1,-1\}$, and by part 3 of Proposition \ref{prop:TwGabidulinSeq}, we have $s^\tau_1(\C')=k+2$ and $s^\tau_1(\C)\geq k+3$. By Lemma \ref{lem:invariant} we conclude that they cannot be equivalent.
Now, recall that $\C, \C'$ are equivalent if and only if there exist $\theta \in \Aut(\Fm)$ and $\A \in \GL_n(q)$ such that $\C'=\theta(\C)\A$. When $\C=\CTGabSheekey{k,\eta}{\autom}{\alphaVec}$ we get $\theta(\C)\A= \CTGabSheekey{k,\theta(\eta)}{\autom}{\theta(\alphaVec)}$.
Assume $\eta'=\theta(\eta)$ for some $\theta \in \Aut(\Fm)$. If $\autom=\tau$ then, since $\alpha_1,\ldots, \alpha_m$ are linearly independent, then so are $\theta(\alpha_1), \ldots \theta(\alpha_m)$ and therefore, $\supp_q(\theta(\alphaVec))=\supp_q(\betaVec)$. This implies that there exists $\A \in \GL_m(q)$ such that $\theta(\alphaVec)\A=\betaVec$ and  $\theta(\C) \A = \C'$.
Hence, for every $\autom$ generator of $\Gal(\Fm/\Fq)$, and for every representative $\eta$ in an orbit in $\mathcal X_q(m,k)$, we have exactly one equivalent class of twisted Gabidulin codes. Moreover, observe that $\CTGabSheekey{k,\eta}{\autom^{-1}}{\betaVec} = \CTGabSheekey{k,\eta^{-1}}{\autom}{\autom^{m-k}(\betaVec)}$. 
This shows that $\C$ and $\C'$ are equivalent if and only if $\tau=\autom$ and $\eta=\theta(\eta')$ for some $\theta \in \Aut(\Fm)$ or $\tau=\autom^{-1}$ and $\eta^{-1}=\theta(\eta')$ for some $\theta \in \Aut(\Fm)$. By counting, we get exactly $\frac{\phi(m)}{2}|\mathcal X_q(m,k)|$ inequivalent twisted Gabidulin codes.
\end{proof}

Analogously to the previous proof one can easily show the known fact that a generalized Gabidulin code is not equivalent to any twisted Gabidulin code with $\eta \neq 0$.

\subsection{The Sequence for Other Twisted Gabidulin Codes}

\begin{theorem}\label{thm:sequence_twisted_codes}
Let $\C := \CTGab{k,\tVec,\hVec,\etaVec}{\autom}{\alphaVec}$ be a (generalized) twisted Gabidulin code with $\numTwists=1$ twist, $\tVec = t \notin \{1,n-k\}$, $\hVec = h \notin \{0,k-1\}$.
Also, let $\alphaVec$ be chosen such that the $\alpha_i$ span a subfield of $\Fqm$, i.e., $\langle \alpha_1,\dots,\alpha_n \rangle = \mathbb{F}_{q^n}$. Then,
\begin{equation*}
s^\autom_{i}(\C) = \begin{cases}
k+1+2i,  \hfill \text{ if }1\leq  i\leq \min\{t-1,n-k-t\}, \\
\min\{k+1+\min\{t-1,n-k-t\} + i, n\}, \, \text{else}.
\end{cases}
\end{equation*}
\end{theorem}

\begin{proof}
By definition, we have
\begin{align*}
\C = \big\langle \alphaVec, \sigma(\alphaVec), \dots,\sigma^{h-1}(\alphaVec), \sigma^{h}(\alphaVec)+\eta \sigma^{k-1+t}(\alphaVec), \\
\sigma^{h+1}(\alphaVec),\dots,\sigma^{k-1}(\alphaVec) \big\rangle, \text{ and} \\
\sigma(\C) = \big\langle \sigma(\alphaVec), \dots,\sigma^{h}(\alphaVec), \sigma^{h+1}(\alphaVec)+\sigma(\eta) \sigma^{k+t}(\alphaVec), \\
\quad \sigma^{h+2}(\alphaVec),\dots,\sigma^{k}(\alphaVec) \big\rangle.
\end{align*}
Since $h \notin \{0,k-1\}$, all vectors $\alphaVec,\sigma(\alphaVec), \dots,\sigma^{k}(\alphaVec)$ are contained in $\mathcal{S}^\autom_{1}(\C) = \C + \sigma(\C)$.
Furthermore, $\sigma^{k-1+t}(\alphaVec)$ and $\sigma^{k+t}(\alphaVec)$ can be obtained by linear combinations of the basis elements above.
Conversely, all basis elements above can be written as linear combinations of these $k+3$ vectors, so we have
$
\mathcal{S}^\autom_{1}(\C) = \big\langle \alphaVec, \dots,\sigma^{k}(\alphaVec),\sigma^{k-1+t}(\alphaVec),\sigma^{k+t}(\alphaVec) \big\rangle.
$
Furthermore, due to $t \notin \{1,n-k\}$, the basis elements are distinct and linearly independent since they are rows of the $\sigma$-Moore matrix $M_{n}^{\sigma}(\alphaVec)$ (due to $k+t < n$).
Hence, $s_1^{\autom}(\C) = \dim \mathcal{S}^\autom_{1}(\C) = k+3$.
The rest of the proof follows inductively.
We claim that $\mathcal{S}^\autom_{i}(\C)$ is given by
\begin{align*}
\mathcal{S}^\autom_{i}(\C) \!=\! \big\langle \alphaVec,\dots,\sigma^{k-1+i}(\alphaVec),\sigma^{k-1+t}(\alphaVec),\dots,\sigma^{k-1+t+i}(\alphaVec) \big\rangle.
\end{align*}
For $i=1$, this is the statement that we proved above.
The induction step follows from part~2 of Proposition~\ref{prop:propertiesSi} and similar arguments as above.

It is left to determine the rank of the given generators of $\mathcal{S}^\autom_{i}(\C)$.
Note that the generators consist of two sets of consecutive rows of the $\sigma$-Moore matrix $M_n^{\sigma}(\alphaVec)$: The rows $\alphaVec,\dots,\sigma^{k-1+i}(\alphaVec)$ and the rows $\sigma^{k-1+t}(\alphaVec),\dots,\sigma^{k-1+t+i}(\alphaVec)$. In between, there are two ``gaps'', i.e., the consecutive rows $\sigma^{k+i}(\alphaVec),\dots,\sigma^{k-2+t}(\alphaVec)$ and $\sigma^{k+t+i}(\alphaVec),\dots,\sigma^{n-1}(\alphaVec)$ are not contained in $\mathcal{S}^\autom_{i}(\C)$.

The dimension of $S^\autom_{i}(\C)$ increases by \emph{two} for $i \to i+1$ as long as these gaps are not closed, since in this case, the vectors $\sigma^{k+i}(\alphaVec)$ and $\sigma^{k+t+i}(\alphaVec)$ are added to the basis of $S^\autom_{i}(\C)$ and are linearly independent of the existing basis elements.

As soon as one of the gaps is closed, the respective ``new vector'' $\sigma^{k+i}(\alphaVec)$ or $\sigma^{k+t+i}(\alphaVec)$ is already in the basis (since the $\alpha_i$ span a subfield of $\Fqm$, we have $\sigma^n(\alphaVec) = \alphaVec$). The first gap is closed for $i \geq t-1$ and the second gap is closed for $i\geq n-k-t$.

Thus, if exactly one gap is closed, which is the case for
\begin{align*}
\min\{t-1,n-k-t\} \leq i < \max\{t-1,n-k-t\},
\end{align*}
then we have $s_{i+1}^{\autom}(\C) = s_{i}^{\autom}(\C)+1$. As soon as both gaps are closed, i.e., for $i \geq \max\{t-1,n-k-t\}$, we have $s_{i+1}(\C) = s_{i}(\C) = n$. This proves the claim.
\end{proof}

\begin{remark}\label{rem:s_sequence_general_alpha_twisted_gab}
For arbitrary $\alphaVec$, the sequence given in Theorem~\ref{thm:sequence_twisted_codes} is not necessarily equal to the given values, but always lower bounded by them. This is due to the fact that for the $\alpha_i$ to span a subfield, we have $\sigma^n(\alphaVec) = \alphaVec$ and in the general case, $\sigma^n(\alphaVec) = \sum_{i=0}^{n-1}\lambda_i \sigma^i(\alphaVec)$ for a non-vanishing linear combination $\lambda_i$ that depends on $\alphaVec$. Hence, it is difficult to analyze whether a vector $\sigma^{i}(\alphaVec)$, for $i \geq n$, contributes a linearly independent vector to the generating set of the code.
\end{remark}

\begin{table*}
\caption{Lower / Upper Bounds on the Number of Inequivalent Twisted Gabidulin Codes (fixed $m,n,k,\alphaVec,\etaVec$, variable $\sigma,\tVec,\hVec$), cf.~Section~\ref{sec:numerical_results}.} 
\label{tab:counting_twisted_codes}
\vspace{-0.4cm}
\begin{center}
\setlength{\tabcolsep}{2pt}
\begin{tabular}{c||c|c|c|c|c|c|c|c|c|c|c|c|c}
$n$ $\backslash$ $k$ 	& 3 & 4 & 5 & 6 & 7 & 8 & 9 & 10 & 11 & 12 & 13 & 14 & 15  \\
\hline \hline
7 &    9 /   36 &    6 /   36 &             &             &             &             &             &             &             &             &             &             &             \\
8 &    7 /   60 &    7 /   64 &    3 /   60 &             &             &             &             &             &             &             &             &             &             \\
9 &   15 /   54 &   21 /   60 &   15 /   60 &    6 /   54 &             &             &             &             &             &             &             &             &             \\
10 &   12 /   84 &   14 /   96 &   13 /  100 &    7 /   96 &    3 /   84 &             &             &             &             &             &             &             &             \\
11 &   40 /  120 &   60 /  140 &   65 /  150 &   60 /  150 &   30 /  140 &   10 /  120 &             &             &             &             &             &             &             \\
12 &   19 /  108 &   18 /  128 &   12 /  140 &   12 /  144 &    8 /  140 &    6 /  128 &    3 /  108 &             &             &             &             &             &             \\
13 &   66 /  180 &  102 /  216 &  108 /  240 &  120 /  252 &   96 /  252 &   72 /  240 &   36 /  216 &   12 /  180 &             &             &             &             &             \\
14 &   30 /  198 &   36 /  240 &   42 /  270 &   39 /  288 &   42 /  294 &   30 /  288 &   21 /  270 &   12 /  240 &    4 /  198 &             &             &             &             \\
15 &   52 /  144 &   64 /  176 &   72 /  200 &   76 /  216 &   72 /  224 &   64 /  224 &   44 /  216 &   36 /  200 &   20 /  176 &    8 /  144 &             &             &             \\
16 &   58 /  312 &   76 /  384 &   92 /  440 &   72 /  480 &   80 /  504 &   72 /  512 &   56 /  504 &   40 /  480 &   28 /  440 &   16 /  384 &    5 /  312 &             &             \\
17 &  136 /  336 &  184 /  416 &  224 /  480 &  240 /  528 &  256 /  560 &  280 /  576 &  240 /  576 &  192 /  560 &  128 /  528 &   96 /  480 &   48 /  416 &   16 /  336 &             \\
18 &   51 /  270 &   54 /  336 &   60 /  390 &   57 /  432 &   60 /  462 &   54 /  480 &   51 /  486 &   39 /  480 &   27 /  462 &   18 /  432 &   12 /  390 &    9 /  336 &    4 /  270 \\

\end{tabular}
\end{center}
\vspace{-0.5cm}
\end{table*}

The following example was given in \cite[Example~7.19]{puchinger2018construction}, where it was stated as an open problem if the given code is equivalent to a Sheekey's twisted Gabidulin code with squared automorphism. With our new criterion we can easily show that this is not the case.

\begin{example}
Consider $3<k<n-5$, $h=k-2$, $t=2$ and the code $\C:= \CTGab{k,\tVec,\hVec,\etaVec}{\autom}{\alphaVec}$, for some $\alpha_i$'s linearly independent over $\Fq$: In this case, $\C$ is generated by 
\begin{equation*}
\left\{\alphaVec,\sigma(\alphaVec),\dots,\sigma^{k-3}(\alphaVec),\sigma^{k-1}(\alphaVec), \eta \sigma^{k+1}(\alphaVec) + \sigma^{k-2}(\alphaVec) \right\},
\end{equation*}
and ${\sigma^2}(\C)$ is generated by
\begin{equation*}
\left\{\sigma^2(\alphaVec),\dots,\sigma^{k-1}(\alphaVec),\sigma^{k+1}(\alphaVec), \sigma^2(\eta) \sigma^{k+3}(\alphaVec) + \sigma^{k}(\alphaVec)\right\}.
\end{equation*}
Similarly, ${\sigma^4}(\C)$ is generated by
\begin{equation*}
\left\{\sigma^4(\alphaVec),\dots,\sigma^{k+1}(\alphaVec),\sigma^{k+3}(\alphaVec), \sigma^4(\eta) \sigma^{k+5}(\alphaVec) + \sigma^{k+2}(\alphaVec)\right\}.
\end{equation*}
Hence, $S_2^{\sigma^2} (\C)$ has generators
\begin{equation*}
\left\{\alphaVec,\dots,\sigma^{k+1}(\alphaVec),\sigma^{k+3}(\alphaVec), \sigma^4(\eta) \sigma^{k+5}(\alphaVec) + \sigma^{k+2}(\alphaVec)\right\}.
\end{equation*}
and we have $s_1^{\sigma^2}(\C) = k+2$ and $s_2^{\sigma^2}(\C) = k+4$.
Note that this holds for any vector $\alphaVec$ with linearly independent entries since, by assumption, $k+5 <n$ (cf.~Remark~\ref{rem:s_sequence_general_alpha_twisted_gab}).
By Proposition~\ref{prop:TwGabidulinSeq}, we know that for Sheekey's twisted Gabidulin code $\C' := \CTGabSheekey{k,\eta'}{(\sigma^2)^j}{\betaVec}$, for some $\betaVec \in \Fqm^n$, we have $s_1^{\sigma^2}(\C') = k+2$ and $s_2^{\sigma^2}(\C') = k+3$. 
Therefore, $\C$ and $\C'$ cannot be equivalent.
\end{example}

\subsection{The Sequence for Gabidulin's New Codes}

We now determine the sequence for Gabidulin's new codes, as defined in 
 Definition \ref{def:codes}. This will show that, if $N(\eta)\neq (-1)^{km}$ (as was assumed in \cite{gabidulin2017new} for the codes to be MRD), they are actually equivalent (and hence equal) to classical Gabidulin codes. 

\begin{theorem}\label{thm:gabidulin_code}
Let $\C := \CGabNew{k,\eta}{\autom,\mathsf{new}}{\alphaVec}$ be one of Gabidulin's new codes as defined in Definition~\ref{def:codes}.
Then, we have
\begin{equation*}
s_i(\C) = \min\{k+i,n\} \quad \forall i = 0,\dots,n.
\end{equation*}
\end{theorem}

\begin{proof}
$\C$ is generated by $\sigma^{i}(\alphaVec)+\sigma^{i}(\eta) \sigma^{k+i}(\alphaVec)$ for $0 \leq i < m-k$ and $\sigma^i(\alphaVec)$ for $m-k \leq i <k$.
$\sigma(\C)$ is generated by
\begin{itemize}
\item $\sigma^{i}(\alphaVec)+\sigma^{i}(\eta) \sigma^{k+i}(\alphaVec)$ for $1 \leq i < m-k$, which are already contained in $\C$,
\item $\sigma^i(\alphaVec)$ for $m-k < i \leq k$, which includes a new, linearly independent, vector $\sigma^k(\alphaVec)$, and
\item $\sigma^{m-k}(\alphaVec)+\sigma^{m-k}(\eta) \sigma^{m}(\alphaVec)$, which is a linear combination of $\sigma^{m-k}(\alphaVec)$ (which is in $\C$), $\sigma^k(\alphaVec)$ (which is in $\sigma(\C)$), and $\sigma^{k}(\alphaVec)+\sigma^{k}(\eta)\sigma^{m}(\alphaVec)$ (which is in $\sigma(\C)$).
\end{itemize}
Hence, their sum, $\mathcal{S}_1^{\sigma}(\C)$ has dimension $k+1$.
The rest of the proof follows inductively by the same arguments.
\end{proof}

\begin{corollary}
Gabidulin's new codes, with $N(\eta)\neq (-1)^{km}$, are equivalent to Gabidulin codes w.r.t.\ the same $\sigma$.
\end{corollary}

\begin{proof}
It was shown in \cite{horlemann2015new} that a linear MRD code of dimension $k$ is equivalent to a Gabidulin code (w.r.t.\ $\sigma$) if and only if $s_1(\C)=k+1$. The claim follows by Theorem~\ref{thm:gabidulin_code}.
\end{proof}

\section{Example and Numerical Results}
\label{sec:numerical_results}

\begin{example}
We determine the sequences $(s_i^{\sigma}(\C))_{i=0}^{n}$ for three examples of Gabidulin, Sheekey's twisted Gabidulin, and twisted Gabidulin codes.
Let $m = 24$, $n=12$, $k=5$, $\ell=1$, $\tVec = t=5$, and $\hVec =h=1$. Also, let the entries of $\alphaVec$ span $\mathbb{F}_{q^{12}}$ and choose $\etaVec = \eta \in \mathbb{F}_{q^{24}} \setminus \mathbb{F}_{q^{12}}$. Then, we have
\begin{align*}
(s_i^{\sigma}(\CGab{5}{\sigma}{\alphaVec}))_{i=0}^{n} &= (5,6,7,8,9,10,11,12,12,\dots,12), \\
(s_i^{\sigma}(\CTGabSheekey{5,\eta}{\sigma}{\alphaVec}))_{i=0}^{n} &= (5,7,8,9,10,11,12,12,\dots,12), \\
(s_i^{\sigma}(\CGab{5,\tVec,\hVec,\etaVec}{\sigma}{\alphaVec}))_{i=0}^{n} &=(5,8,10,11,12,12,\dots,12),
\end{align*}
so the three codes are inequivalent since their sequences differ.
\end{example}

Table~\ref{tab:counting_twisted_codes} presents lower bounds on the number of inequivalent twisted Gabidulin codes with $\ell=1$ twist for all lengths $n \leq 18$ and dimensions
 $2<k<n-2$, where the $\alpha_i$ span a field $\Fqn$ and we choose the smallest field, $\Fqm := \mathbb{F}_{q^{2n}}$, for which there is a sufficient condition for the codes to be MRD.
Furthermore, we fix an $\eta \in \Fqm \setminus \Fqn$.

Each table cell contains two values, $a$ / $b$. The number $a$ is a lower bound on the number of inequivalent twisted codes with $\ell=1$, $h$, $t$, and $\sigma = \cdot^{q^s}$ with
\begin{align*}
(s,h,t) \in \mathcal{I}_{m,n,k} := \{(s,h,t) \, : \, &1 \leq s < n, \, \gcd(s,m)=1, \\
&0\leq h <k, \, 1\leq t \leq n-k\}.
\end{align*}
The bound $a$ is derived by computing the sequences $\{s_{i}^{\sigma'}\}$, where $\sigma'$ ranges over all generators of the Galois group $\Gal(\Fqn/\Fq)$.
The number $b$ is defined as the maximal number of inequivalent codes $b = |\mathcal{I}_{m,n,k}|$.

For most of the presented parameters, a large fraction (up to almost $1/2$) of the parameters in $\mathcal{I}_{m,n,k}$ result in inequivalent codes. E.g., for $n=17$ and $k=8$, there are $576$ valid choices of $s$, $h$, and $t$. Among these choices, there are exactly $280$ codes with pairwise disjoint sequences ${s_{i}^{\sigma'}}$. Hence, almost half of the codes are pairwise inequivalent.

Intuitively, we expect that at most about $1/2$ (or even less) of all codes are inequivalent due to symmetries in $t$, $h$, and $\sigma$ yielding equivalent codes for different parameter choices (similarly to the symmetry argument in the proof of Theorem~\ref{thm:numberGabidulin}). Hence, it seems that some of the lower bounds in Table~\ref{tab:counting_twisted_codes} are close to being tight, but this must be investigated in more detail by studying symmetries in the code parameters.

\bibliographystyle{IEEEtran}
\bibliography{main}

\end{document}